\documentclass[runningheads]{llncs}

\usepackage{latexsym}
\usepackage{amsfonts,amssymb,amsmath}
\usepackage{graphicx, gastex}
\usepackage{cite}
\usepackage{mathrsfs}
\usepackage{color}

\DeclareSymbolFont{rsfscript}{OMS}{rsfs}{m}{n}
\DeclareSymbolFontAlphabet{\mathrsfs}{rsfscript}
\DeclareMathOperator{\dt}{.}
\DeclareMathOperator{\df}{df}
\DeclareMathOperator{\Syn}{Syn}
\DeclareMathOperator{\Cong}{Cong}

\DeclareMathOperator{\Fix}{Fix}

\DeclareMathOperator{\Suff}{Suff}
\DeclareMathOperator{\Fact}{Fact}
\DeclareMathOperator{\Pref}{Pref}

\newcommand{\rf}{\rightarrow}

\newcommand{\la}{\langle}
\newcommand{\ra}{\rangle}

\newcommand{\col}{\textcolor{red}}

\begin{document}
\title{Representation of (Left) Ideal Regular Languages
by Synchronizing Automata}
\author{Marina Maslennikova\inst{1}, Emanuele Rodaro\inst{2}}
\authorrunning{M. Maslennikova, E. Rodaro}
\titlerunning{Representation of Regular Languages by Synchronizing Automata}

\institute{Institute of Mathematics and Computer Science\\Ural Federal University, Ekaterinburg, Russia\\ \and Centro de Matem\'{a}tica, Faculdade de Ci\^{e}ncias\\
        Universidade do Porto, 4169-007 Porto, Portugal\\
\email{maslennikova.marina@gmail.com, emanuele.rodaro@fc.up.pt}}

\maketitle

\begin{abstract}
We follow language theoretic approach to synchronizing automata and \v{C}ern\'{y}'s conjecture initiated in a series of recent papers. We find  a precise lower bound for the reset complexity of a principal ideal languages. Also we show a strict connection between principal left ideals and synchronizing automata. We characterize regular languages whose minimal deterministic finite automaton is synchronizing and possesses a reset word belonging to the recognized language. \\
\textbf{Keywords:} ideal language, synchronizing automaton, reset word, reset complexity, reset left regular decomposition, strongly connected automaton.
\end{abstract}

\section*{Introduction}
Let $\mathscr{A}=\langle Q,\Sigma,\delta\rangle$ be a \textit{deterministic finite automaton} (DFA),
where $Q$ is the \textit{state set}, $\Sigma$ stands for the \textit{input alphabet},
and $\delta: Q\times\Sigma\rightarrow Q$ is the totally defined \textit{transition function} defining
the action of the letters in $\Sigma$ on $Q$. The function $\delta$ is extended uniquely to a function $Q\times \Sigma^*\rightarrow Q$, where $\Sigma^*$ stands for the free monoid over $\Sigma.$ The latter function is still denoted by $\delta.$
In the theory of formal languages the definition of a DFA usually includes the \textit{initial state} $q_0\in Q$ and the set $F\subseteq Q$ of \textit{terminal states}. In this case a DFA is defined as a quintuple $\mathscr{A}=\langle Q,\Sigma,\delta,q_0,F\rangle.$ We will use this definition when dealing with automata as devices for recognizing languages. A language $L\subseteq \Sigma^*$ is said to be \emph{recognized} (or \emph{accepted}) by an automaton $\mathscr{A}=\langle Q, \Sigma, \delta, q_0, F\rangle$ if $L=\{w\in \Sigma^*\mid \delta(q_0,w)\in F\}$, in this case we put $L=L[\mathscr{A}]$.
We also use standard concepts of the theory of formal languages such as regular language, minimal automaton etc.~\cite{Perrin}

A language $I\subseteq\Sigma^*$ is called a \emph{two-sided ideal} (or simply an \emph{ideal}) if $I$ is non-empty and $\Sigma^*I\Sigma^*\subseteq I$. A language $I\subseteq\Sigma^*$ is called a \emph{left} (respectively, \emph{right}) \emph{ideal} if $I$ is non-empty and $\Sigma^*I\subseteq I$ (respectively, $I\Sigma^*\subseteq I$). In what follows we will consider only languages which are regular, thus we will drop the term ``regular'' and henceforth a given language will be implicitly a regular language. If it is said ``ideal language'' or simply ``ideal'', it means that exactly a two-sided ideal language is considered, otherwise it will be explicitly mentioned which class of languages we are focusing on.

A DFA $\mathscr{A}=\langle Q,\Sigma,\delta\rangle$ is called
\emph{synchronizing} if there exists a word~$w\in\Sigma^{*}$ whose action leaves
the automaton in one particular state no matter at which state in $Q$ it is applied, i.e. $\delta(q,w)=\delta(q',w)$
for all $q, q' \in Q$. Any word with this property is said to be
\emph{reset} for the DFA $\mathscr{A}$.
For the last 50 years synchronizing automata received a great deal of attention.
For a brief introduction to the theory of synchronizing automata we refer the reader
to the survey~\cite{Vo_Survey}.

Recently in a series of papers~\cite{SOFSEM, PrincIdFI, ReRo13, DCFS14} a language theoretic (and descriptional complexity) approach to the study of synchronizing automata has been developed. In the present paper we continue to study synchronizing automata from a language theoretic point of view and find a new approach to the \v{C}ern\'{y} conjecture in this way.
We denote by $\Syn(\mathrsfs{A})$ the language of reset words for a given synchronizing automaton $\mathscr{A}$.
It is well known that $\Syn(\mathrsfs{A})$ is a regular language~\cite{Vo_Survey}. Furthermore, it is an ideal in $\Sigma^*$, i.e. $\Syn(\mathrsfs{A})=\Sigma^{*}\Syn(\mathrsfs{A})\Sigma^{*}$.
On the other hand, every ideal language $I$ serves as the language of reset words for some automaton. For instance, the minimal automaton recognizing $I$ is synchronized by $I$~\cite{SOFSEM}. Thus synchronizing automata can be considered as a special representation of ideal languages. The complexity of such a representation is measured by the \emph{reset complexity} $rc(I)$ which is the minimal possible number of states in a synchronizing automaton $\mathscr{A}$ such that $\Syn(\mathscr{A})=I$.
Every such automaton $\mathscr{A}$ is called \textit{minimal synchronizing automaton} (for brevity, MSA). Let $sc(I)$ be the \emph{state complexity} of $I,$ i.e. the number of states in the minimal automaton recognizing $I$.
Since the minimal automaton recognizing $I$ has $I$ as the language of reset words, we clearly have $rc(I)\leq sc(I)$.
Moreover, there are ideals $I_n$ for every $n \geq 3$ such that $rc(I_n) = n$ and $sc(I_n) = 2^n - n$, see~\cite{SOFSEM}.
So representation of an ideal language by means of one of its MSAs can be exponentially more succinct than
its ``traditional'' representation via minimal automaton.
However, no reasonable algorithm is known for computing an MSA for a given language.
One of the obstacles is that MSA is not uniquely defined. Furthermore, the problem of checking, whether a given synchronizing automaton with at least five letters is an MSA for a given ideal language, has recently been shown to be \textbf{PSPACE}-complete~\cite{DCFS14}.

Another source of motivation for studying representations of ideal languages by means
of synchronizing automata comes from the famous \v{C}ern\'{y}'s conjecture~\cite{Ce64}. In 1964
\v{C}ern\'{y} constructed for each $n>1$ a synchronizing $n$-state automaton $\mathscr{C}_n$ whose shortest reset word has length $(n-1)^2.$  Later \v{C}ern\'{y}
conjectured that those automata represent the worst possible case, that is,
every synchronizing automaton with $n$ states possesses a reset word of length at most $(n-1)^2$.
Despite intensive efforts of researchers, this conjecture still remains open. One can
restate easily the
\v{C}ern\'{y} conjecture
in terms of reset complexity.
Let $||I||$ be the minimal length of words in
an ideal language $I.$ The \v{C}ern\'{y} conjecture holds true if and only if $rc(I)\ge \sqrt{||I||}+1$ for every ideal $I$.
The latter inequality would provide the desired quadratic upper bound on the length of the shortest reset word of a synchronizing automaton.

Thus, a deeper study of reset complexity may help to shed light on this longstanding conjecture. In this language theoretic approach to the \v{C}ern\'{y} conjecture, strongly connected synchronizing automata play an important role. Recall that a DFA is called \emph{strongly connected} if for each pair of different states $(p,q)$ there exists a word mapping $p$ to $q.$ It is well known that the \v{C}ern\'{y} conjecture holds true whenever
it holds true for strongly connected automata~\cite{Vo_CIAA07}. In this regard, an interesting question was posed in \cite{PrincIdFI}. The question concerns
the problem of finding a strongly connected synchronizing automaton whose set of reset words is equal to a given ideal language. Indeed, while the minimal automaton recognizing an ideal language $I$ is always a synchronizing automaton with a unique \emph{sink} state (i.e. a state fixed by all letters), finding examples of strongly connected synchronizing automata $\mathrsfs{A}$ with $\Syn(\mathrsfs{A})=I$ is a non-trivial task.
In~\cite{ReRo13} it is proved that such strongly connected automaton always exists for an ideal over alphabet of size at least two. The construction itself is non-trivial and rather technical. Furthermore, the upper bound on the number of states of the associated strongly connected automaton is a double exponential.
The approach of \cite{ReRo13} has the extra advantage of detaching the \v{C}ern\'{y} conjecture from the automata point of view. This is achieved by introducing a purely language theoretic notion of \emph{reset left regular decomposition} of an ideal. This notion will be recalled in Section~\ref{sec: prel}. Here we just focus on the connection between these decompositions and the \v{C}ern\'{y} conjecture. Given an ideal $I$, the size of the smallest reset left regular decomposition of $I$ is denoted by $rdc(I).$ This value can be viewed as the number of states of the smallest strongly connected synchronizing automaton $\mathrsfs{A}$ with $\Syn(\mathrsfs{A})=I$.
It is clear that $rc(I)\le rdc(I)$ and we have
\begin{theorem}\cite[Theorem 6]{ReRoTech}\label{theo: Cerny rdc reform}
\v{C}ern\'{y}'s conjecture holds if and only if for any ideal $I$ we have $rdc(I)\ge \sqrt{||I||}+1$.
\end{theorem}
Therefore,
the importance of the studies of issues like finding more effective constructions of reset left regular decompositions (or equivalently their associated automata) is evident.

Another interesting observation is the following. For each $n\geq 3$ the corresponding MSA's for the aforementioned ideals $I_n$ (with $rc(I_n)$ and $sc(I_n)=2^n-n$) turned out to be strongly connected.
Thus one may expect that there always exists a strongly connected MSA for an ideal language. However, in~\cite{GuMasPribFG} it has been shown that a strongly connected MSA for a given ideal language does not always exist. Moreover, there are ideals $J_n$ for every $n \geq 3$ such that $rc(J_n) = n+1$ and $rdc(J_n) = 2^n$. Thus the smallest strongly connected automaton having a given ideal language $I$ as the
language of reset words may be exponentially larger than an MSA for $I$.




Recall that an ideal $I$ is called \emph{finitely generated} if $I=\Sigma^*U\Sigma^*$ for some finite set $U\subseteq\Sigma^*$.
Such languages have been viewed as languages of reset words of synchronizing automata in \cite{PriRoMinWords,PriRo11}. Note that the aforementioned languages $J_n$ are finitely generated ideals. In~\cite{PrincIdFI} it is considered the partial case of \emph{principal} ideal languages, i.e. languages of the form $\Sigma^*w\Sigma^*,$ for some $w\in\Sigma^{*}$. If $|w|$ denotes the length of $w\in\Sigma^*$, then we have

\begin{theorem}[\cite{PrincIdFI}]\label{theorem: princ}
For the language $\Sigma^*w\Sigma^*,$ there is a strongly connected automaton $\mathrsfs{B}$ with $|w|+1$ states, such that
$\Syn(\mathrsfs{B})=\Sigma^*w\Sigma^*$. Such an automaton can be constructed in $O(|w|^2)$ time.
\end{theorem}
In the present paper we enforce the previous result by showing that the automaton $\mathrsfs{B}$ from Theorem~\ref{theorem: princ} is actually an MSA for a given language. More precisely, we prove that $rdc(I)=rc(I)= \|I\|+1$, for every principal ideal language $I.$ In particular, this result solves an open question posed in~\cite{PrincIdFI} regarding the size of the minimal strongly connected synchronizing automaton for which a given principal ideal language serves as the language of reset words.
We show that \emph{principal left ideals}, i.e. ideals of the form $\Sigma^{*}w$ for some word $w$, play also a fundamental role in \v{C}ern\'{y}'s conjecture. Indeed, we characterize strongly connected synchronizing  automata via homomorphic images of automata belonging to a particular class $\mathcal{L}(\Sigma)$ of automata.
The class $\mathcal{L}(\Sigma)$ is formed by all the trim automata $\mathrsfs{A}=\langle Q,\Sigma,\delta,q_0, \{q_{0}\}\rangle$ such that $L[\mathrsfs{A}]=w^{-1}\Sigma^{*}w$ for some word $w\in\Sigma^{*}$. In Section~\ref{sec: lower principal} we reduce Cern\'{y}'s conjecture to the same conjecture for the quotients of automata from the class $\mathcal{L}(\Sigma)$.
In view of this connection we study automata recognizing languages of the form $w^{-1}\Sigma^{*}w$ for some $w\in\Sigma^{*}$.
We provide a compact formula to calculate the syntactic complexity of a language $I=w^{-1}\Sigma^*w$. This value is defined just by the length of $w$ and by the quantity of distinct prefixes, suffixes and factors in $w$.
Another interesting feature of such languages concerns the construction of the minimal automaton $\mathscr{A}_w$ recognizing the language $w^{-1}\Sigma^*w$. It turns out that $w\in\Syn(\mathscr{A}_w).$ Thus, in this context, we have that a word of the language recognized by the automaton is also a reset word for this automaton. Hence it is quite natural to ask in which cases the minimal automaton recognizing a given regular language $L$ is synchronized by some word from $L$. Here we answer this question by proving a criterion for the minimal automaton recognizing $L$ to be synchronized by some word from $L$. We state this criterion in terms of the notion of a constant of $L$ introduced by Sch\"{u}tzenberger~\cite{Schutz}.
The notion of a constant is widely studied and finds applications in bioinformatics and coding theory \cite{BFZ,SynchSymp}. 

\section{Preliminaries}\label{sec: prel}
Let $\mathrsfs{A}=\la Q,\Sigma,\delta,q_{0},F\ra$ be a deterministic finite automaton. The corresponding triple $\la Q,\Sigma,\delta\ra,$ where the initial state and the set of final states are deliberately omitted, is called the \emph{underlying semiautomaton} of $\mathscr{A}.$ If the transition function $\delta$ is clear from the context, we will write $q\dt w$ instead of $\delta(q,w)$ for $q\in Q$ and $w\in\Sigma^{*}$. This notation extends naturally to any subset $H\subseteq Q$ by putting $H\dt w=\{\delta(q,w) \mid q\in H\}$. A DFA $\mathrsfs{A}=\la Q,\Sigma,\delta,q_{0},F\ra$ is called \emph{trim} whenever each state $q\in Q$ is reachable from $q_0$ and each state $t\in F$ is reachable from some state $q\in Q$.

In our context a (automaton) \emph{homomorphism} $\varphi:\mathrsfs{A}\rf\mathrsfs{B}$  between the DFAs $\mathrsfs{A}=\la Q,\Sigma,\delta\ra$ and $\mathrsfs{B}=\la T,\Sigma,\xi\ra$ is a map $\varphi:Q\rf T$ preserving the action of letters, i.e. $\varphi(\delta(q,a))=\xi(\varphi(q),a)$ for all $a\in\Sigma$. Note that $\varphi(Q)$ identifies a sub-automaton of $\mathrsfs{B}$ denoted by $\varphi(\mathrsfs{A})$, and we say that $\varphi(\mathrsfs{A})$ is a \emph{homomorphic image} of $\mathrsfs{A}$.
A binary relation $\rho\subseteq Q\times Q$ is a \emph{congruence} for the automaton $\mathrsfs{A}=\la Q,\Sigma,\delta,q_{0},F\ra$ if  $(q_{1},q_{2})\in\rho $ implies $(\delta(q_{1},u), \delta(q_{2},u))\in\rho$ for all $u\in\Sigma^{*}$, $q_{1},q_{2}\in Q$. The \emph{quotient automaton} of a DFA $\mathscr{A}$ with respect to a congruence $\rho$ is denoted by $\mathrsfs{A}/\rho=\la Q/\rho,\Sigma,\delta',[q_{0}],F/\rho\ra$, where $[q]$ denotes the $\rho$-class containing $q$, and the transition function $\delta':Q/\rho\times \Sigma\rightarrow Q/\rho$ is defined be the rule $\delta'([q],u)=[\delta(q,u)]$, for all $u\in\Sigma^*$, $q\in Q$. We denote by $\Cong(\mathrsfs{A})$ the set of all the congruences of the DFA $\mathrsfs{A}$, the \emph{index} of a congruence $\rho\in\Cong(\mathrsfs{A})$ is the cardinality of the state set of $\mathrsfs{A}/\rho$. For any integer $k$, we use the symbol $\Cong_{k}(\mathrsfs{A})$ to denote the (possibly empty) set of congruences on $\mathrsfs{A}$ of index $k$.

Denote the $i$-th letter of a word $w\in \Sigma^+$ by $w[i]$ and the prefix $w[1]w[2]\ldots w[i]$ by $w[1..i]$. For indices $1\le i<j\le |w|$ we use the notation $w[i..j]$ to indicate the factor $w[i]w[i+1]\ldots
w[j]$. If $1\le i<j$ then we put $w[j..i]=\varepsilon$. For $u,w\in \Sigma^{*}$ we say that $u$ is a \emph{prefix}, (\emph{suffix} or \emph{factor}, respectively) of $w$ if $w=uu_2$ ($w=u_1u$ or $w=u_1uu_2$, respectively) for some $u_1,u_2\in\Sigma^{*}$. We also write $u\le _{p}w$ ($u\le_{s}w$ or $u\le_{f}w$, respectively) if $u$ is a prefix (suffix or a factor of $w$, respectively). We write $u<_{p}w$ ($u<_{s}w$ or $u<_{f}w$) if $u$ is a proper prefix (suffix or factor, respectively) of $w$. For a given language $L\subseteq \Sigma^*$ and $w\in\Sigma^*$ we put $Lw=\{xw\mid x\in L\}$, $wL=\{wx\mid x\in L\}$. The \emph{left} (\emph{right}) \emph{quotient} of $L$ by a word $w$ is the set $w^{-1}L=\{v\in\Sigma^{*}:wv\in L\}$ ($Lw^{-1}=\{v\in\Sigma^{*}:vw\in L\}$).
We recall the following definition from~\cite{ReRo13}:
\begin{definition}\label{defn: regular dec}
  A reset left regular decomposition is a collection $\{I_i\}_{i\in F}$ of
  disjoint left ideals $I_i$ on $\Sigma^*$, for some finite set $F$, satisfying the following two conditions.  \begin{itemize}
  \item[i)] For any $a\in\Sigma$ and $i\in F$, there is an index $j\in F$
    such that $I_ia\subseteq I_j$.
  \item[ii)]  Let $I=\uplus_{i\in F} I_{i}$. For any $u\in\Sigma^*$ if
    there is an $i\in F$ such that $Iu\subseteq I_i $, then $u\in I$.
  \end{itemize}
\end{definition}
Denote by $\textbf{RLD}_{\Sigma}$ the class of the reset left regular decompositions over $\Sigma$. The notation $\textbf{SCSA}_{\Sigma}$ stands for the class of all strongly connected synchronizing automata over $\Sigma$. In~\cite{ReRo13} it has been shown that an ideal language $I$ is strongly connected if and only if it has a reset left regular decomposition.
The proof of this statement provides a bijection between the classes $\textbf{RLD}_{\Sigma}$ and $\textbf{SCSA}_{\Sigma}$. This fact was stated in the following theorem.
\begin{theorem}[Theorem 4,~\cite{ReRo13}]\label{theo: characterization 2}
The map $\mathcal{A}:\textbf{RLD}_{\Sigma}\rightarrow\textbf{SCSA}_{\Sigma}$ defined by the rule
$$
  \mathcal{A}:\{I_i\}_{i\in F}\mapsto \mathcal{A}(\{I_i\}_{i\in F})=\la \{I_i\}_{i\in F},\Sigma,\eta\ra
$$
with $\eta(I_i,a)=I_j$ for $a\in\Sigma$ if and only if  $I_{i}a\subseteq I_{j}$ is a bijection with inverse given by $\mathcal{I}:\textbf{SCSA}_{\Sigma}\rf\textbf{RLD}_{\Sigma}$ defined by the rule
$$
  \mathcal{I}:\mathrsfs{B}=\la Q,\Sigma,\delta\ra\mapsto \{I_q\}_{q\in Q}=\{\{u\in\Sigma^*: \delta(Q,u)=q\}\}_{q\in Q}.
$$
\end{theorem}

\section{Lower bounds for the reset complexity of principal ideal languages}\label{sec: lower principal}
In this section we prove that $rdc(I)=rc(I)\ge n+1$ for a principal ideal language $I=\Sigma^{*}w\Sigma^{*}$ with $|w|=n$. First we recall some auxiliary facts and definitions from~\cite{PriRoMinWords}. Let us consider an automaton $\mathrsfs{A}=\langle Q,\Sigma,\delta\rangle$. For a word $u\in\Sigma^{*}$, the \emph{maximal fixed set} $m(u)$ is the largest subset of $Q$ fixed by $u$, i.e. $m(u)\dt u=m(u)$. Note that $m(u)=Q\dt u^{k(u)}$ for some minimal integer $k(u)$ and it is not difficult to see that $k(u)\le |Q|-|m(u)|$ (see~\cite[Lemma 2]{PriRoMinWords}).
A synchronizing DFA $\mathscr{A}=\langle Q,\Sigma,\delta\rangle$ is called \emph{finitely generated} if the language $\Syn(\mathscr{A})$ is a finitely generated ideal. The following theorem is proved using the same technique of~\cite[Theorem 4]{PriRoMinWords}, for the sake of completeness we present the proof in the appendix.
\begin{theorem}\label{Theo: form of synch word}
Let ${\mathrsfs A}=\langle
Q,\Sigma,\delta\rangle$ be a finitely generated synchronizing automaton with $|Q|=n$. Then for any word $v\in\Sigma^{+}$ we have that either $v^{k(v)}\in\Syn(\mathscr{A})$, or there is a word $\tau$ with $|\tau|\le n-1$, such that $v^{k(v)}\tau v^{k(v)}\in\Syn(\mathscr{A})$.
\end{theorem}

We are now in position to prove the main theorem of this section.
\begin{theorem}\label{theo: princ rc}
Let $I=\Sigma^{*}w\Sigma^{*}$ be a principal ideal language, then $rc(I)=|w|+1$.
\end{theorem}
\begin{proof}
Since in~\cite[Lemma 1]{SOFSEM} it has been shown that $rc(I)=|w|+1$ for $w=a^{n}$, we may assume $|\Sigma|>1$. By Theorem~\ref{theorem: princ} we have $rc(I)\leq |w|+1$. Suppose, contrary to our claim, that there is a synchronizing automaton ${\mathrsfs A}=\la Q,\Sigma,\delta\ra$ with $|Q|= n\le |w|$ for which $I$ serves as the language of reset words. The equality $|w|=1$ implies that $rc(I)=2$, so in what follows we assume that $|w|>1$. Let $a$ and $b$ be the initial and final letter of $w$ respectively. Denote by $a^{r}$ the maximal prefix of $w$ of the form $a^{l}$, $l\in\mathbb{N}$, and by $b^{h}$ the maximal suffix of $w$ of the form $b^{l}$, $l\in\mathbb{N}$. We consider the following cases.

\textbf{Case 1.} Assume $a\neq b$. Thus $w$ can be factorized as $w=a^{r}ub^{h}$ for some $u\in\Sigma^{*}$. Suppose first that $u\in\Sigma^{+}$. Let us take $v=a^{|w|}b^{|w|}$. By Theorem~\ref{Theo: form of synch word} we have two cases: either $v^{k(v)}\in \Syn(\mathrsfs{A})=I$, or there is a word $\tau$ with $|\tau|\le n-1\le |w|-1$ such that $v^{k(v)}\tau v^{k(v)}\in I$.

Suppose that $v^{k(v)}\in I$. Thus $w\le_{f}v^{k(v)}$, and since $w$ can not be a factor of either $a^{|w|}$ or $b^{|w|}$, it must be a factor of $v$. Since $u\neq\varepsilon$ we have that $u[1]\neq a$ and $u[|u|]\neq b$ by the definition of $a^{r},b^{h}$. Thus $w$ is not a factor of $v$, a contradiction. Therefore, we can assume that $v^{k(v)}\tau v^{k(v)}\in I$, and so $w\le_{f}v^{k(v)}\tau v^{k(v)}$. From the arguments above we have that $w$ can not be a factor of $v$ or $v^{k(v)}$, so we have $w\le_f v\tau v$. Since $w$ is not a factor of $v$, $w[1]=a\neq b$, $w[|w|]=b\neq a$, we obtain $w\le_{f}\tau$. Hence $|w|\le |\tau|\le |w|-1$, which is a contradiction.

Hence we may consider $u=\varepsilon$, and so $w=a^{r}b^{h}$. In~\cite[Lemma 1]{SOFSEM} it was shown that $rc(I)=|w|+1$ for $w\in\{a^{n},b^{n}\}$. In the same paper it was obtained that $rc(I)=|w|+1$ for $w=a^{n-1}b$, thus we can assume that $r\ge 1$ and $h\ge 2$. If $r>1$ we take $v=a^{r-1}b^{h-1}$. By Theorem \ref{Theo: form of synch word} we have that either $v^{k(v)}\in I$, or $v^{k(v)}\tau v^{k(v)}\in I$ for some word $\tau$ with $|\tau|\le n-1\le |w|-1$. Obviously, $w=a^{r}b^{h}$ can not be a factor of $v^{k(v)}$. Therefore, $w$ is a factor of $v\tau v$. Again using simple technique from combinatorics on words it is easy to see that $w$ must be a factor of $\tau$. Hence we get $|w|\le |\tau|\le |w|-1$, a contradiction. If $r=1$ we take $v=ab^{h-1}$. By Theorem \ref{Theo: form of synch word} we have that either $v^{k(v)}\in I$, or $v^{k(v)}\tau v^{k(v)}\in I$ for some word $\tau$ with $|\tau|\le n-1\le |w|-1$. The word $w=ab^{h}$ is not a factor of $v^{k(v)}$, thus $w\le_{f} v\tau v$. Note that $h>2$, hence $w$ must be a factor of $\tau$, which is again a contradiction.

\textbf{Case 2.} Assume $a=b$. If $w\in\{a^{n},b^{n}\}$ then $rc(I)=|w|+1$~\cite[Lemma 1]{SOFSEM}. Therefore, we can assume that $w=a^{r}ua^{h}$ for some $u\in\Sigma^{+}$ with $u[1]\neq a$, $u[|u|]\neq a$. In this case we apply Theorem~\ref{Theo: form of synch word} with $v=b$ for some $b\in\Sigma\setminus\{a\}$. Providing the same arguments as above, it is easy to prove that $w$ has to be a factor of a word $\tau$ with $|\tau|\le |w|-1$, which again leads to the contradiction $|w|\le |\tau|\le |w|-1$.\qed
\end{proof}

Note that by Theorem~\ref{theorem: princ} we have the equality $rc(I)=rdc(I)=|w|+1$.

\section{A lifting property for strongly connected synchronizing automata}\label{sec: lifting prop}
The aim of this section is to prove that strongly connected synchronizing automata are all and only all the homomorphic images of automata from some particular class.
\begin{definition}\label{defn: trim aut}
The considered class $\mathcal{L}(\Sigma)$ is formed by all the trim automata $\mathrsfs{A}=\langle Q,\Sigma,\delta,q_0, \{q_{0}\}\rangle$ such that $L[\mathrsfs{A}]=w^{-1}\Sigma^{*}w$ for some word $w\in\Sigma^{*}$.
\end{definition}
Here we reduce Cern\'{y}'s conjecture to the same conjecture for the quotients of automata from the class $\mathcal{L}(\Sigma)$. We have the following proposition.
\begin{proposition}\label{prop: class property}
Let $\mathrsfs{A}\in \mathcal{L}(\Sigma)$ with $L[\mathrsfs{A}]=w^{-1}\Sigma^{*}w$. Then $\mathrsfs{A}$ is a strongly connected synchronizing automaton and $w$ is a reset word for $\mathrsfs{A}$.
\end{proposition}
\begin{proof}
Since $\mathrsfs{A}=\la Q,\Sigma,\delta,q_0, \{q_{0}\}\ra$ is a trim DFA, for each $q\in Q$ there is a word $u\in\Sigma^{*}$ such that $q_{0}\dt u=q$. On the other hand, $uw\in w^{-1}\Sigma^{*}w=L[\mathrsfs{A}]$, thus we have $q_{0}=q_{0}\dt uw=q\dt w$. In this way, we obtain that $q\dt w=q_{0}$ for each $q\in Q$, i.e. $w\in\Syn(\mathscr{A})$.

Now we prove that $\mathrsfs{A}$ is a strongly connected DFA. Take two arbitrary states $q_{1}, q_{2}\in Q$. Since $\mathrsfs{A}$ is a trim DFA there is a word $u$ such that $q_{0}\dt u=q_{2}$. Thus, since $q_{1}\dt w=q_{0}$, we have $q_{1}\dt(wu)=q_{0}\dt u=q_{2}$.\qed
\end{proof}
Let $w,u\in\Sigma^{*}$, we denote by  $u\wedge_{s} w$ the maximal suffix of the word $u$ that appears in $w$ as a prefix. We have the following lemma (for the proof see appendix).

\begin{lemma}\label{lem: charact wedge}
For any $u,v,w\in\Sigma^{*}$, $(uv)\wedge_{s} w=((u\wedge_{s} w)v)\wedge_{s}w$. Furthermore, for any $v$ with $|v|\ge w$, $(uv)\wedge_{s} w=v\wedge_{s}w$.
\end{lemma}

Let $\mathscr{A}=\la Q,\Sigma,\delta,q_0,F\ra$ be a DFA. For a state $q\in Q$ we define the \emph{right} language of $q$  $L_q[\mathscr{A}]=\{u\in\Sigma^*\mid q\dt u\in F\}$.
For $p,q\in Q$ we say that $p$ and $q$ are \emph{equivalent} if $L_q[\mathscr{A}]=L_p[\mathscr{A}]$. A DFA with a distinguished initial state and distinguished set of final states is minimal if it contains no (different) equivalent states and all states are reachable from the initial state.
The automata from $\mathcal{L}(\Sigma)$ recognize languages which are left quotients of the form $w^{-1}\Sigma^{*}w$. In fact these languages are recognized by automata with exactly $|w|+1$ states as it is shown in the following proposition.
\begin{proposition}\label{prop: language charact. A_w}
  Consider the automaton $\mathrsfs{A}_w=\la P(w),\Sigma,\xi,q_n, \{q_n\}\ra$ where $P(w)=\{q_{0},\dots, q_{n}\}$ is the set of prefixes of the word $w$ of length $0\le i\le |w|=n$, and the transition function is defined by the rule $\xi(q_{i},a)=(q_{i}a)\wedge_{s} w$ for all $a\in\Sigma$, $q_i\in P(w)$. The DFA $\mathrsfs{A}_w$ is the minimal automaton recognizing the language
  \begin{equation}\label{eq: language A_w}
    L[\mathrsfs{A}_w]=w^{-1}\Sigma^*w
  \end{equation}
\end{proposition}
\begin{proof}
By Lemma~\ref{lem: charact wedge} it is straightforward to see that $\xi(q_{i},u)=(q_{i}u)\wedge_{s}w$ for all $u\in\Sigma^*$, $q\in Q$. First we prove the equality (\ref{eq: language A_w}). Let $u\in\Sigma^*$ and $\xi(q_n,u)= q_n$. Hence $w=q_n=(wu)\wedge_{s}w$, i.e. $wu\in\Sigma^{*}w$. Conversely, if $u\in w^{-1}\Sigma^*w$, that is $wu\in\Sigma^*w$, then $(wu)\wedge_{s}w=w=q_{n}$. This implies that $\xi(q_n,u)=q_n$, i.e. $u\in L[\mathrsfs{A}_w]$.

We now consider the minimality issue. We verify that each state $q_i\in P(w)$ is reachable from the initial state $q_n$. Indeed, let $a$ be any letter from $\Sigma$ different from $w[1]$. We have the equality $\xi(q_n,a^n)=q_0$. The word $w[1..i]$ maps $q_0$ to $q_i$, so we have $\xi(q_n,a^nw[1..i])=q_i$. Now we take any $q_i, q_j\in P(w)$ with $i\neq j$. Without loss of generality we can assume $i<j$. Consider the word $u=w[j+1,n]$. We have $\xi(q_j,u)=q_n$ while $\xi(q_i,u)\neq q_n$ since $|q_iu|<|w|$. Hence $q_i,q_j$ are not equivalent.
So the DFA $\mathrsfs{A}_w$ is minimal.\qed
\end{proof}

\begin{example}
Take $w=aba$, $\Sigma=\{a,b\}$. The minimal automaton $A_w$ recognizing the language $L=w^{-1}\Sigma^*w$ is shown in Fig.~\ref{fig: aba}.
\end{example}

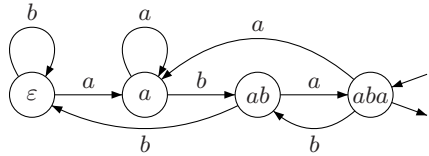
\begin{figure}[ht]
\unitlength=1.5mm
\begin{center}
\begin{picture}(30,9)
\gasset{Nh=4,Nw=4,Nmr=2,loopdiam=4}
\node(A)(0,5){$\varepsilon$}
\node(B)(10,5){$a$}
\node(C)(20,5){$ab$}
\node(D)(30,5){$aba$}
\imark[iangle=20](D)
\fmark[fangle=-20](D)
\drawedge(A,B){$a$}
\drawedge(B,C){$b$}
\drawedge(C,D){$a$}
\drawedge[curvedepth=3](C,A){$b$}
\drawedge[curvedepth=-5,ELside=r](D,B){$a$}
\drawedge[curvedepth=3](D,C){$b$}
\drawloop[loopangle=90](A){$b$}
\drawloop[loopangle=90](B){$a$}
\end{picture}
\caption{Automaton $\mathscr{A}_w$ for $w=aba$}
\label{fig: aba}
\end{center}
\end{figure}

Let $\mathrsfs{A}=\la Q,\Sigma,\delta\ra$ be a strongly connected synchronizing automaton.
By Theorem~\ref{theo: characterization 2} we can build for $\mathrsfs{A}$ the associated reset left regular decomposition $\mathcal{I}(\mathrsfs{A})=\{I_{i}\}_{i\in Q}$ where $\uplus_{i\in Q}I_{i}=I=\Syn(\mathrsfs{A})$. Take a word $w\in I$ of minimum length. Let $\sigma_{w}$ be a binary relation on $I$ defined as follows. For $u,v\in I$ we say that
\begin{equation}\label{eq: sigma w}
(u,v)\in\sigma_w \mbox{ if and only if } u,v\in I_{i} \mbox{ for some } i\in Q \mbox{ and } u\wedge_{s}w=v\wedge_{s}w
\end{equation}
We have the following lemma.
\begin{lemma}\label{lem: prop sigma_w}
Let $\mathrsfs{A}=\la Q,\Sigma,\delta\ra$ be a strongly connected DFA, and $\{I_{i}\}_{i\in Q}$ its associated reset left regular decomposition.
The relation $\sigma_w$ is a right congruence on $I$. Furthermore, each $\sigma_w$-class is a left ideal contained in $I_{i}$ for some $i\in Q$.
\end{lemma}
\begin{proof}See appendix.\end{proof}

Note that $\mathrsfs{A}_w\in\mathcal{L}(\Sigma)$.
Now we are in position to state the main result of this section.
\begin{theorem}\label{theo: lifting property}
Let $\mathrsfs{A}=\la Q,\Sigma,\delta\ra$ be a strongly connected synchronizing automaton. For any reset word $w$ of minimum length, there is a DFA $\mathrsfs{B}\in\mathcal{L}(\Sigma)$ with $L[\mathrsfs{B}]=w^{-1}\Sigma^*w$ and
$$
\Sigma^{*}w\Sigma^{*}\subseteq \Syn(\mathrsfs{B})\subseteq \Syn(\mathrsfs{A})
$$
such that $\mathrsfs{A}$ is a homomorphic image of $\mathrsfs{B}$.
\end{theorem}
\begin{proof}See appendix.\end{proof}

Note that the previous theorem is constructive and we can effectively compute the lifted automaton $\mathrsfs{B}$ of the statement. Moreover, the minimality of the length of the word $w$ among the reset words is also necessary to ensure the fact that each equivalence class is a left ideal. We have the following corollary.
\begin{corollary}
The class of strongly connected synchronizing automata are all and only all the homomorphic images of the class  $\mathcal{L}(\Sigma)$ formed by the trim automata $\mathrsfs{A}=\la Q,\Sigma,\delta,\{q_0\}, q_{0}\ra$ such that $L[\mathrsfs{A}]=w^{-1}\Sigma^{*}w$ for some word $w\in\Sigma^{*}$.
\end{corollary}
\begin{proof}
By Proposition \ref{prop: class property} we have that any $\mathrsfs{A}\in\mathcal{L}(\Sigma)$ is a strongly connected synchronizing automata, hence any homomorphic image $\varphi(\mathrsfs{A})$ is also a strongly connected synchronizing automaton. On the other hand, by Theorem \ref{theo: lifting property} any strongly connected synchronizing automaton is a homomorphic image of a DFA from $\mathcal{L}(\Sigma)$.\qed
\end{proof}
Using Theorem \ref{theo: lifting property} we can give another reformulation of Cerny's conjecture using the automata from $\mathcal{L}(\Sigma)$.
\begin{theorem}\label{theo: another restatement}
Cerny's conjecture holds if and only if for any $\mathrsfs{B}\in\mathcal{L}(\Sigma)$  and $\rho\in\Cong_{k}(\mathrsfs{B})$ for all $k<\sqrt{\| \Syn(\mathrsfs{B})\|}+1$ we have
$$
\|\Syn(\mathrsfs{B}/\rho)\|<\|\Syn(\mathrsfs{B})\|
$$
\end{theorem}
\begin{proof}
See appendix.
\end{proof}

\section{Some properties of the automaton $A_w$}
In view of the results of the previous section, left quotients of principal left ideals seem to play a fundamental role in the \v{C}ern\'{y} conjecture. In this regard we initiate a study of automata recognizing languages of the form $w^{-1}\Sigma^*w$. In this section we provide a compact formula to calculate the size of the syntactic semigroup of a language $I=w^{-1}\Sigma^*w$, $w\in\Sigma^*$.

For a regular language $L\subseteq \Sigma^{*}$ the \emph{Myhill conguence}~\cite{Myhill} $\thickapprox_L$ of $L$ is defined as follows:$$u\thickapprox_L \text{ if and only if }xuy\in L \Leftrightarrow xvy\in L \text{ for all }x,y\in\Sigma^{*}.$$ This congruence is also known as \emph{the syntactic congruence} of L. The quotient semigroup $\Sigma^{+}/\thickapprox_L$ of the relation $\thickapprox_L$ is called the \emph{syntactic semigroup} of L. The syntactic semigroup of $L$ is known to be isomorphic to the transition semigroup of the minimal DFA recognizing $L$. The \textit{syntactic complexity} $\sigma(L)$ of a regular language $L$ is the cardinality of its syntactic semigroup. The notion of syntactic complexity is studied quite extensively: for a survey of this topic we refer the reader to \cite{HolK3}.
Also the notion of the syntactic semigroup finds interesting application in the theory of synchronizing automata. Indeed, let $I$ be an ideal language, $\mathcal{S}$ the syntactic semigroup of $I$ and $\mathcal{S}(\mathrsfs{B})$ the transition semigroup of a synchronizing DFA $\mathrsfs{B}$ for which $I=\Syn(\mathscr{B})$. In~\cite{PrincIdFI} it has been shown that $\mathcal{S}$ is a homomorphic image of $\mathcal{S}(\mathrsfs{B})$.

Recall that $u\in\Sigma^{+}$ is an \textit{inner factor} of $w$ if there exist words $x,y\in\Sigma^{+}$ such that $w=xuy$. Denote by $\Fact(w)$ the set of different inner factors of $w$, by $\Suff(w)$ the set of proper non-empty suffixes of $w$ which do not appear in $w$ as inner factors, by $\Pref(w)$ the set of proper non-empty prefixes of $w$ which do not appear in $w$ as suffixes or inner factors, by $\Pref_{syn}(w)$ the set of prefixes of $w$ synchronizing $\mathscr{A}_w$. We have the following
\begin{proposition}\label{prop: synt comp quotient}
Let $I=w^{-1}\Sigma^*w$ for some $w\in\Sigma^*$. The syntactic complexity of $I$ is equal to $$\sigma(I)=|w|+1+|\Pref(w)|+|\Fact(w)|+|\Suff(w)|-|\Pref_{syn}(w)|.$$
\end{proposition}
\begin{proof}See appendix.\end{proof}

Note that by Proposition~\ref{prop: synt comp quotient} we get an effective algorithm to calculate the syntactic complexity of the left quotient $w^{-1}I$ by $w$ of a principal left ideal $I=\Sigma^*w$.

By Proposition~\ref{prop: class property} the minimal automaton $\mathscr{A}_w$ recognizing $I=w^{-1}\Sigma^*w$ is strongly connected and $w\in\Syn(\mathscr{A}_w)$. Further we show that $\mathscr{A}_w$ is finitely generated. Recall that a reset word $w$ for a given synchronizing DFA $\mathscr{A}$ is called \emph{minimal} if none of its proper prefixes nor suffixes belong to $\Syn(\mathscr{A})$. Denote by $\Syn_{min}(\mathscr{A}_w)$ the set of all minimal reset words for a given synchronizing DFA $\mathscr{A}_w$.

\begin{proposition}
For each $w\in\Sigma^*$, $\mathscr{A}_w$ is a finitely generated synchronizing automaton.
\end{proposition}
\begin{proof}
In order to obtain the desired result we prove that the set $\Syn_{min}(\mathscr{A}_w)$ is finite. Take an arbitrary $u\in\Syn_{min}(\mathscr{A}_w)$. If $|u|>|w|$ then $u$ is not minimal. Indeed, by the definition of the transition function of $\mathscr{A}_w$ and by Lemma~\ref{lem: charact wedge} we get, for all $q_i\in P(w)$, $q_i\dt u=q_iu\wedge_s w=u\wedge_s w=u[2..|u|]\wedge_s w$ since $|u|-1\ge |w|$. Thus we have $|u|\le|w|$. However, there is just finite amount of words of length at most $|w|$. Hence $\mathscr{A}_w$ is a finitely generated synchronizing automaton.\qed
\end{proof}

\section{Representation of regular languages by synchronizing automata}\label{sec: repr}

In this section $\mathscr{A}_L$ stands for the minimal DFA recognizing a regular language $L$. In some cases $\mathscr{A}_L$ may have a unique \emph{non-accepting} sink state $s$, i.e. $s\not\in F$. It may turn out that $\mathscr{A}_L$ is synchronizing and, therefore, each reset word brings the whole automaton to $s$. If this is not the case one may consider partial synchronization in the following sense. A DFA $\mathscr{A}=\la Q,\Sigma,\delta,q_0,F\ra$ with a non-accepting sink state $s$ is called \emph{partially synchronizing} if there exists a word $w\in\Sigma^*$ such that $Q\dt w=\{s,q\}$ for some state $q\in Q$. Any word with this property is said to be \emph{partial reset} word for the DFA $\mathscr{A}$. And the set of all partial reset words for $\mathscr{A}$ is denoted by $\Syn^{par}(\mathscr{A})$.

Let $L$ be a regular language. If $L$ is an ideal language then $\mathscr{A}_L$ is synchronizing and $\Syn(\mathscr{A}_L)=L$. In Section~\ref{sec: lifting prop} it has been shown that the minimal automaton recognizing the language $w^{-1}\Sigma^*w$ is synchronizing and $w$ is a reset word for this automaton. On the other hand, $w\in w^{-1}\Sigma^*w$. So in this case we have that the minimal automaton recognizing a given language $L$ is synchronizing and some word from $L$ is also a reset word for the automaton. In this regard the following interesting question arises. How to describe all regular languages $L$ for which $\mathscr{A}_L$ is synchronizing and $L\cap\Syn(\mathscr{A}_L)\neq\emptyset$?
In this section we answer this question.

Let $L\subseteq\Sigma^*$ be a regular language. A word $w\in\Sigma^*$ is a \emph{constant} for $L$ if the implication $$u_1wu_2\in L, u_3wu_4\in L\Rightarrow u_1wu_4\in L$$ holds for all $u_1,u_2,u_3,u_4\in\Sigma^*.$
We denote the set of all constants of $L$ by $C(L)$. Note that the set $C(L)$ contains the ideal $Z(L)=\{w\mid \Sigma^*w\Sigma^*\cap L=\emptyset\}$. Constant words of a regular language $L$ satisfy the following property, also stayed in~\cite{Schutz}.
\begin{lemma}\label{lem: const}
Let $L\subseteq\Sigma^*$ be a regular language and let $\mathscr{A}_L$ be the minimal automaton recognizing $L$ with set of states $Q$. If $\mathscr{A}_L$ has a non-accepting sink state $s$ then a word $w\in\Sigma^*$ is a constant for $L$ if and only if $|Q\dt w|\leq 2$. If $\mathscr{A}_L$ does not have a non-accepting sink state $s$ then a word $w\in\Sigma^*$ is a constant for $L$ if and only if $|Q\dt w|=1$.
\end{lemma}

By this lemma it follows that constants of a regular language $L$ are described precisely via reset and partial reset words of the minimal automaton recognizing~$L$.
Let $L\subseteq\Sigma^*$, denote by $\overline{L}$ the \emph{complement} to $L$, that is $\overline{L}=\Sigma^*\setminus L$.
\begin{proposition}\label{prop: nonempty complete}
The automaton $\mathscr{A}_L$ is synchronizing and $L\cap\Syn(\mathscr{A}_L)\neq\emptyset$ if and only if the following properties hold:

(i) $C(L)\neq\emptyset$

(ii) $\overline{L}$ does not contain right ideals.
\end{proposition}
\begin{proof}
Consider the DFA $\mathscr{A}_L=\la Q,\Sigma,\delta,q_0,F\ra$.
Assume that $\mathscr{A}_L$ is synchronizing and the condition $L\cap\Syn(\mathscr{A}_L)\neq\emptyset$ holds. We take any $w\in\overline{L}\cap\Syn(\mathscr{A}_L)$. By Lemma~\ref{lem: const} we have $w\in C(L)$. Arguing by contradiction assume that $\overline{L}$ contains a right ideal. This means that there is a strongly connected component $H\subseteq Q\setminus F$ without outgoing transitions leading to $F$. Thus, for all $w\in\Syn(\mathscr{A}_L)$, we have $H\dt w \cap F=\emptyset$, hence $L\cap\Syn(\mathscr{A}_L)=\emptyset$, which is a contradiction.

Assume that properties (i) and (ii) hold. By property (ii) $\mathscr{A}_L$ does not have a non-accepting sink state. Thus, by Lemma~\ref{lem: const} each constant of $L$ is a reset word for $\mathscr{A}_L$, and since $C(L)$ is not empty, $\mathscr{A}_L$ is synchronizing. Arguing by contradiction, assume that $L\cap\Syn(\mathscr{A}_L)=\emptyset$, hence $\Syn(\mathscr{A}_L)\subseteq \overline{L}$. However, the language $\Syn(\mathscr{A}_L)$ is a right ideal, a contradiction.\qed
\end{proof}

The following proposition deals with the complementary case.

\begin{proposition}\label{prop: empty complete}
The automaton $\mathscr{A}_L$ is synchronizing and $L\cap\Syn(\mathscr{A}_L)=\emptyset$ if and only if the following properties hold:

(i) $Z(L)\neq\emptyset$

(ii) $\overline{L}$ contains a right ideal.
\end{proposition}
\begin{proof}
Consider the DFA $\mathscr{A}_L=\la Q,\Sigma,\delta,q_0,F\ra$.
Assume that $\mathscr{A}_L$ is synchronizing and the condition $L\cap\Syn(\mathscr{A}_L)=\emptyset$ holds. Arguing by contradiction assume that $\overline{L}$ does not contain a right ideal. By Proposition~\ref{prop: nonempty complete} we get that $L\cap\Syn(\mathscr{A}_L)\neq\emptyset$, which is a contradiction. So property (ii) holds. This property is equivalent to the existence of a strongly connected component $H\subseteq Q\setminus F$. By the minimality of $\mathscr{A}_L$ we obtain $|H|=1$, thus $H$ contains just a non-accepting sink state. Since $\mathscr{A}_L$ is synchronizing, each $w\in\Syn(\mathscr{A}_L)$ brings the whole DFA $\mathscr{A}_L$ to $s$, hence $Z(L)\neq\emptyset$.

Conversely, assume that properties (i) and (ii) hold. Again, by property (ii) there is a non-accepting sink state in $\mathscr{A}_L$. Thus each $w$ from $Z(L)$ is a reset word for $\mathscr{A}_L$. Arguing by contradiction, assume that $L\cap\Syn(\mathscr{A}_L)\neq\emptyset$. Thus by Proposition~\ref{prop: nonempty complete} $\overline{L}$ does not contain right ideals. Contradiction.\qed
\end{proof}

Note that in order to check whether property (ii) in both of the previous propositions is satisfied, it is enough to check whether there is a strongly connected component in $Q\setminus F$. The latter can be implemented in time $O(n\cdot|\Sigma|)$, where $n=|Q|$.
Note that some problems related two constants of languages are considered in~\cite{Berl}. In particular, the problem of deciding whether a given partial 2-letter automaton is partially synchronizing is shown to be $NP$-complete (the action of the transition function on some states of a given automaton may be undefined). The notion of a partial synchronizing word from~\cite{Berl} is defined analogously to the notion of partial reset word here.
Now we formally state the following CONSTANT problem:

-- \emph{Input:} a regular language $L$ over $\Sigma$, presented via its minimal recognizing DFA $\mathscr{A}_L$.

-- \emph{Question:} is it true that $C(L)\neq \emptyset$?

We can suppose that $\mathscr{A}_L$ has a non-accepting sink state $s$, since otherwise the problem is equivalent to testing $\mathscr{A}_L$ for synchronization in usual sense.
First we prove the following
\begin{lemma}\label{lem: const characterization}
Let $\mathscr{A}_L=\la Q,\Sigma,\delta, q_0,F\ra$ have a non-accepting sink state $s$. The set $C(L)$ is not empty if and only if for each pair $\{p,q\}$ of different states $p,q\in Q$ there is a word $u$ such $\{p,q\}\dt u\subseteq \{s,r\}$ for some $r\in Q$.
\end{lemma}
\begin{proof}
Clearly, if $C(L)\neq\emptyset$ the desired property holds by Lemma~\ref{lem: const}. Conversely, take any pair $\{p,q\}$ of different states, then there is a word $w_1\in\Sigma^*$ such that $\{p,q\}\dt w_1\subseteq\{s,r\}$ for some $r\in Q$. We clearly have $|Q\dt w_1|<|Q|$. Consider now the set $Q\dt w_1$. If $|Q\dt w_1|\leq 2$ then $w_1\in C(L)$, so we are done. Otherwise, if $|Q\dt w_1|>2$ then take again any two different states $p',q'\in Q\dt w_1$ such that $p',q'\neq s$. Hence there is a word $w_2\in \Sigma^*$ such that $\{p',q'\}\dt w_2\subseteq \{s,r'\}$ for some $r'\in Q$. We have the inequality $|Q\dt w_1w_2|<|Q\dt w_1|<|Q|$. Consider now the set $Q\dt w_1w_2$. If $|Q\dt w_1w_2|\leq 2$ then $w_1w_2\in C(L)$, so we are done. Arguing by induction we get, through a finite number of steps as described above,  a word $w$ such that $|Q\dt w|\le 2$. That is $w\in C(L)$.\qed
\end{proof}

Recall that for a given DFA $\mathscr{A}=\la Q,\Sigma,\delta,q_0,F\ra$ the \emph{power automaton} $\mathcal{P}(\mathscr{A})$ is constructed as follows. Its state set $\mathcal{Q}$ includes all non-empty subsets of $Q$ and the transition function is a natural extension of $\delta$ on the set $\mathcal{Q}\times\Sigma$. The latter function is still denoted by $\delta$. Denote by $\mathcal{P}^{[2]}(\mathscr{A})$ the subautomaton of the power automaton $\mathcal{P}(\mathscr{A})$ consisting only of 2-element and 1-element subsets of $Q$.
\begin{proposition}\label{prop: const alg}
CONSTANT can be solved in time $O(n^5\cdot|\Sigma|)$, where $n=|Q|$.
\end{proposition}
\begin{proof}
We use Lemma~\ref{lem: const characterization} to establish nonemptiness of the set $C(L)$. First we build the corresponding automaton $\mathcal{P}^{[2]}(\mathscr{A})$ that can be done in time $O(n^2\cdot|\Sigma|)$. This automaton has $\frac{n(n+1)}{2}$ states. Take any pair $\{p,q\}$ of different states $p,q\in Q$, $p,q\neq s$. Take any pair $\{r,s\}$, $r\neq s$. We put $L_{p,q,r,s}=\{w\mid \{p,q\}\dt w=\{r,s\}\}$, $L_{p,q,r}=\{w\mid \{p,q\}\dt w=\{r\}\}$, $L_{p,q,s}=\{w\mid \{p,q\}\dt w=\{s\}\}$. Nonemptiness of any of these three sets can be checked in time $O(n^2\cdot|\Sigma|)$ by a breadth first search in $\mathcal{P}^{[2]}(\mathscr{A})$. The latter may be done for all possible pairs $\{p,q\}$ and $\{r,s\}$ (in the worst case). Since there are $\frac{n(n-1)^2}{2}$ possible choices for the pairs $\{p,q\}$ and $\{r,s\}$, we get a cost of $O(n^5\cdot|\Sigma|)$. Finally, we obtain that it takes $O(n^5\cdot|\Sigma|)$ time to solve CONSTANT.
\end{proof}

\textbf{Remark.} Some partial results of the paper have been presented on the Third Russian Finnish Symposium on Discrete Mathematics RuFiDiM2014. The conference provided local proceedings (not indexed) in which we have presented an extended abstract of the communication without any proof.

\section*{Acknowledgements}
The first author acknowledges support from the Presidential Programme for young researchers, grant MK-3160.2014.1, from the Presidential Programme ``Leading Scientific Schools of the Russian Federation'', project no.\ 5161.2014.1, and from the Russian Foundation for Basic Research, project no.\ 13-01-00852.
The last author acknowledges support from the European Regional Development Fund through the
programme COMPETE and by the Portuguese Government through the FCT -- Funda\c c\~ao para a Ci\^encia e a Tecnologia under the project PEst-C/MAT/UI0144/2013 as well as support from the FCT project SFRH/BPD/65428/2009.

\newpage
\textbf{Appendix}

Let us consider an automaton $\mathrsfs{A}=\langle Q,\Sigma,\delta\rangle$. For a subset $S\subseteq Q$, denote by $\Fix(S)$ the set of words $u\in\Sigma^{*}$ such that $S\dt u=S$, and by $\Syn(S)$ the set of words $\{u\in\Sigma^*$ such that $\mid |S\dt u|=1\}$. A subset $S\subseteq Q$ is called \emph{reachable} if $Q\dt u=S$ for some $u\in\Sigma^*$.
We shall use the equality $m(u^{\ell})=m(u)$ for any $\ell\ge 1$ and $u\in\Sigma^*$ (see \cite[Lemma 3]{PriRoMinWords}). The class of finitely generated synchronizing automata has the following combinatorial characterization.
\begin{theorem}\cite[Theorem 1]{PriRoMinWords}\label{theo: characterization}
A synchronizing automaton $\mathrsfs{A}=\langle Q,\Sigma,\delta\rangle$ is finitely generated if and only if, for any reachable subset $S\subseteq Q$ with $1<|S|<|Q|$ and for any $u\in\Fix(S)$ it holds that $\Syn(S)=\Syn(m(u))$.
\end{theorem}
The \emph{deficiency} of a word $u\in\Sigma^*$ with respect to $\mathrsfs{A}$ is the difference $\df(w)=|Q|-|Q\dt w|$. We make of use the following result from~\cite{MPV04}.
\begin{theorem}\label{Theo: definciency}
Given a  synchronizing automaton ${\mathrsfs A}=\langle
Q,\Sigma,\delta\rangle$ and the words $u,v\in\Sigma^+$ such that
$\df(u)=\df(v)=k>1$, there exists a word $\tau$, with $|\tau|\le
k+1$, such that $\df(u\tau v)>k$.
\end{theorem}
Now we are in position to prove Theorem~\ref{Theo: form of synch word}.

\noindent \textbf{Theorem~\ref{Theo: form of synch word}.}
\emph{Let ${\mathrsfs A}=\langle
Q,\Sigma,\delta\rangle$ be a finitely generated synchronizing automaton with $|Q|=n$. Then for every word $v\in\Sigma^{+}$ we have that either the word $v^{k(v)}$ is reset for $\mathscr{A}$, or there is a word $\tau$ with $|\tau|\le n-1$, such that $v^{k(v)}\tau v^{k(v)}$ is a reset word for $\mathscr{A}$.}

\noindent \begin{proof}
Let us take an arbitrary word $v\in\Sigma^{+}$. If $|m(v)|=1$ then $v^{k(v)}\in\Syn(\mathscr{A})$, so we are done. Now we may assume that $|m(v)|>1$. We construct the following set
$$
\textsc{Reach}(v)=\{S\subseteq m(v)\mid S=m(v)\dt u,
u\in\Sigma^*, |S|>1\}
$$
which is non-empty since $m(v)\in\textsc{Reach}(v)$. By Theorem \ref{theo: characterization} for any $S\in\textsc{Reach}(v)$ it holds that
\begin{equation}\label{eq:Prop_Reach}
\Syn(S)= \Syn(m(v)).
\end{equation}
Indeed, since $S\subseteq m(v)$ we have
$v^{\ell}\in \Fix(S)$ for some integer $\ell>1$. On the other hand, synchronizing DFA $\mathrsfs{A}$ is finitely generated. Thus applying
Theorem \ref{theo: characterization} we get $ \Syn(S)= \Syn(m(v^{\ell}))= \Syn(m(v))$. Now let $H=Q\dt v^{k(v)}u$ be an element of $\textsc{Reach}(v)$ of
minimal cardinality and let $k'=n-|H|=\df(v^{k(v)}u)$. Since
$|H|>1$ we have $k'\le n-2$. Since $\mathrsfs{A}$ is
synchronizing we have by Theorem~\ref{Theo: definciency} that there is a word $\tau$
with $|\tau|\le k'+1\le n-1$ such that $\df(v^{k(v)}u\tau v^{k(v)}u)>k'$,
i.e.\ $|Q\dt v^{k(v)}u\tau v^{k(v)}u|<n-k'=|H|$. We claim that the
word $v^{k(v)}u\tau v^{k(v)}uv^{k(v)}$ is reset for $\mathscr{A}$. Indeed,  $Q\dt v^{k(v)}u\tau
v^{k(v)}u\subseteq Q$, thus $Q\dt v^{k(v)}u\tau
v^{k(v)}uv^{k(v)}\subseteq Q\dt v^{k(v)}=m(v)$. So we obtain that $Q\dt v^{k(v)}{u}\tau
v^{k(v)}{u}v^{k(v)}$ is an element of $\textsc{Reach}(v)$ with
$$
|Q\dt v^{k(v)}u\tau a^{k(v)}uv^{k(v)}|\le|Q\dt v^{k(v)}u\tau
v^{k(v)}u|<|H|.
$$
Hence by the choice of $H$ we get $|Q\dt v^{k(v)}u\tau v^{k(v)}uv^{k(v)}|=1$, that is the word $v^{k(v)}u\tau v^{k(v)}uv^{k(v)}$ is reset for $\mathscr{A}$. In fact even the word $v^{k(v)}u\tau v^{k(v)}$ is reset for $\mathscr{A}$. Indeed, consider
the set $S=Q\dt v^{k(v)}u\tau v^{k(v)}$. Let us assume that $|S|>1$. In this case it holds that $S\in
\textsc{Reach}(v)$, hence by (\ref{eq:Prop_Reach}) we obtain $uv^{k(v)}\in \Syn(S)= \Syn(m(v))$, so
$$1=|m(v)\dt
uv^{k(v)}|=|Q\dt v^{k(v)}uv^{k(v)}|=|H\dt v^{k(v)}|.$$
But by the choice of $H$ we have the inequality $|H|>1$. Furthermore, $H\subseteq m(v)$. On the other hand, $v$ acts as a permutation on $m(v)$. Therefore, we have $|H\dt v^{k(v)}|=|H|>1$, which is a contradiction and we get $|S|=1$. Thus, since $S=Q\dt v^{k(v)}u\tau v^{k(v)}=H\dt \tau v^{k(v)}$, by (\ref{eq:Prop_Reach}) we obtain
$$
\tau v^{k(v)}\in \Syn(Q\dt v^{k(v)}u)= \Syn(H)= \Syn(m(v))= \Syn(Q\dt v^{k(v)}),
$$
i.e. $v^{k(v)}\tau v^{k(v)}$ is a reset word for $\mathscr{A}$.\qed
\end{proof}

\textbf{Lemma~\ref{lem: charact wedge}.}
\emph{For any $u,v,w\in\Sigma^{*}$, $(uv)\wedge_{s} w=((u\wedge_{s} w)v)\wedge_{s}w$. Furthermore, for any $v$ with $|v|\ge w$, $(uv)\wedge_{s} w=v\wedge_{s}w$.}

\begin{proof}
Let $t=(uv)\wedge_{s} w$. If $t<_{s} v$, then it is easy to see that $t=hv\wedge_{s}w$ where $h$ is an arbitrary suffix of $u$. In particular, we have $t=((u\wedge_{s} w)v)\wedge_{s}w$. Thus we can assume that $t<_{s}uv$ and there is a non-empty word $r\in\Sigma^{+}$ such that $r\le _{s}u$, $r\le_{p}w$ and $t=rv$. Hence $r\le_{s}(u\wedge_{s} w)$. Since $t$ is the maximal suffix of $uv$ which is also a prefix of $w$ and $r\le_{s}(u\wedge_{s} w)\le_{s}u$ we get that $t$ is also the maximal suffix of $(u\wedge_{s} w)v$ which is also a prefix of $w$, i.e. $t=((u\wedge_{s} w)v)\wedge_{s}w$. The last statement of the lemma follows trivially from the definition.\qed
\end{proof}

\textbf{Lemma~\ref{lem: prop sigma_w}.}
\emph{Let $\mathrsfs{A}=\la Q,\Sigma,\delta\ra$ be a strongly connected DFA, and $\{I_{i}\}_{i\in Q}$ its associated reset left regular decomposition.
The relation $\sigma_w$ is a right congruence on $I$. Furthermore, each $\sigma_w$-class is a left ideal contained in $I_{i}$ for some $i\in Q$.
}
\begin{proof}
Clearly, $\sigma_w$ is an equivalence relation on $I$. Let $a\in\Sigma$ and $(u,v)\in\sigma_w$, i.e. $u,v\in I_i$ for some $i\in Q$ and $u\wedge_{s}w=v\wedge_{s}w$. By property $i)$ of Definition \ref{defn: regular dec} we have $ua,va\in I_{i}a\subseteq I_{j}$ for some $j\in Q$. Furthermore, by Lemma \ref{lem: charact wedge} and $u\wedge_{s}w=v\wedge_{s}w$ we have
$$
(ua)\wedge_{s}w=((u\wedge_{s} w)a)\wedge_{s}w=((v\wedge_{s} w)a)\wedge_{s}w=(va)\wedge_{s}w
$$
hence $(ua,va)\in\sigma_{w}$.
Since the number of possible prefixes of $w$ is finite, by the definition of $\wedge_{s}$ we have that $\sigma_{w}$ has finite index. Take any $u\in I$, denote by $[u]$ the $\sigma_w$-class of $u$. Clearly, $[u]\subseteq I_i$ for some $i\in Q$. Since $w$ is a reset word for $\mathscr{A}$ of minimum length, for each word $u\in I$ we have $|u|\ge |w|$, thus for each $v\in\Sigma^*$ we have $(vu)\wedge_s w=u\wedge_s w$ by Lemma~\ref{lem: charact wedge}. Hence [u] is a left ideal in $\Sigma^*$ contained in $I_{i}$ for some $i\in Q$.\qed
\end{proof}

\textbf{Theorem~\ref{theo: lifting property}.}
\emph{
Let $\mathrsfs{A}=\la Q,\Sigma,\delta\ra$ be a strongly connected synchronizing automaton. For any reset word $w$ of minimum length, there is a DFA $\mathrsfs{B}\in\mathcal{L}(\Sigma)$ with $L[\mathrsfs{B}]=w^{-1}\Sigma^*w$ and
$$
\Sigma^{*}w\Sigma^{*}\subseteq \Syn(\mathrsfs{B})\subseteq \Syn(\mathrsfs{A})
$$
such that $\mathrsfs{A}$ is a homomorphic image of $\mathrsfs{B}$.}

\begin{proof}
Let $\mathrsfs{A}=\la Q,\Sigma,\delta\ra$ be a strongly connected synchronizing automaton.
By Theorem~\ref{theo: characterization 2} we can build for $\mathrsfs{A}$ the associated reset left regular decomposition $\mathcal{I}(\mathrsfs{A})=\{I_{i}\}_{i\in Q}$ where $\uplus_{i\in Q}I_{i}=I=\Syn(\mathrsfs{A})$. Since $w\in I$, there is some $j\in Q$ such that $w\in I_j$ and thus $\Sigma^{*}w\subseteq I_{j}$. Let $\sigma_{w}$ be a binary relation on $I$ defined by~(\ref{eq: sigma w}). By Lemma~\ref{lem: prop sigma_w} each $\sigma_w$-class is a left ideal contained in some $I_i$ for some $i\in Q$.

Therefore $\sigma_{w}$ induces a refinement $\{J_{t}\}_{t\in T}$ of $\{I_{i}\}_{i\in Q}$ for some set of indices $T=\{v_{0},\ldots,v_{m}\}$. Since $\sigma_{w}$ is a right congruence, for any $v_{i}\in T, a\in\Sigma$ we have $J_{v_{i}}a\subseteq J_{v_{j}}$ for some  $T=\{v_0,\ldots,v_m\}$. Thus $\Sigma$ defines an action $\lambda$ on $T$ defined by $\lambda(v_{i},a)=v_{h}$ where $v_{h}$ is the unique index of $T$ such that $J_{v_{i}}a\subseteq J_{v_{h}}$. Using a simple induction on the length of the words it is straightforward to check that the following condition holds
\begin{equation}\label{eq: lambda ext}
\lambda(v_{i},u)=v_{t},\:u\in\Sigma^{*}\mbox{ iff } J_{v_{i}}u\subseteq J_{v_{t}}
\end{equation}
Note that $\Sigma^{*}w$ is a $\sigma_w$-class belonging to $\{J_{t}\}_{t\in T}$, say $\Sigma^*w=J_{v_{0}}$. Therefore, consider the DFA $\mathrsfs{B}=\la H,\Sigma,\lambda, v_{0},\{v_{0}\}\ra$ where
$$
H=\{v_{j}\in T:v_{j}=\lambda(v_{0},u)\mbox{ for some }u\in\Sigma^{*}\}
$$
and let us prove that $\mathrsfs{B}\in\mathcal{L}(\Sigma)$ with $L[\mathrsfs{B}]=w^{-1}\Sigma^*w$. Since $J_{v_{h}}w\subseteq\Sigma^{*}w=J_{v_{0}}$ for any $v_{h}\in H$, then by (\ref{eq: lambda ext}) we have $\lambda(v_h,w)=v_0$, so $\mathrsfs{B}$ is a trim DFA. We now prove the equality $L[\mathrsfs{B}]=w^{-1}\Sigma^*w$. Let $u\in\Sigma^{*}$ such that $\lambda(v_{0}, u)=v_{0}$, by (\ref{eq: lambda ext}) this is equivalent to $\Sigma^{*}wu\subseteq \Sigma^{*}w$, and it is not difficult to see that this is also equivalent to $wu\wedge_{s}w=w$. In the proof of Proposition \ref{prop: language charact. A_w} we have seen that $wu\wedge_{s}w=w$ is equivalent to $\xi(q_{n},u)=q_{n}$, i.e. $u\in L[\mathrsfs{A}_w]=w^{-1}\Sigma^*w$.

The first inclusion in the statement of the theorem $\Sigma^{*}w\Sigma^{*}\subseteq \Syn(\mathrsfs{B})$ is a consequence of Proposition \ref{prop: class property}. Let us prove the second inclusion. The following claim is of use.
\begin{claim}
For any $I_{j}$ with $j\in Q$ there is at least a $\sigma_w$-class $J_{v_{h}}$ such that $J_{v_{h}}\subseteq I_{j}$ for some $v_{h}\in H$.
\end{claim}
\begin{proof}
Indeed, this property clearly holds for the left ideal $I_{i}$ containing $J_{v_{0}}=\Sigma^{*}w$. Thus consider any $I_{j}$ for $j\neq i$. Since $I_{j}$ is a left ideal, for any $u\in I_{j}$ we get $I_{i}u\subseteq I_{j}$. In particular we get $J_{v_{0}}u\subseteq I_{i}u\subseteq I_{j}$.\qed
\end{proof}
Take any $u\in\Syn(\mathscr{B})$, thus there exists some $v_i\in T$ such that $J_{v_{k}}u\subseteq J_{{v_{i}}}$ for all $v_k\in T$. Using the Claim we conclude that there exists some $i\in Q$ such that $I_ju\subseteq I_i$ for all $j\in Q$\col{.} Therefore, $Iu\subseteq I_i$ and by condition ii) of Definition~\ref{defn: regular dec} we obtain $u\in \Syn(\mathscr{A})$.

Let us prove the last statement of the theorem. Consider the map $\varphi: H\rf Q$ defined by $\varphi(v_{h})=j$ where $j\in Q$ is the unique index such that $J_{v_{h}}\subseteq I_{j}$. We claim that $\varphi:\mathrsfs{B}\rf\mathrsfs{A}$ is a homomorphism. Indeed, take any $h\in H$, $a\in\Sigma$, and put $t=\lambda(h,a)$, $r=\varphi(t)$, $q=\varphi(h)$. Since $J_h\subseteq I_q$, $J_{h}a\subseteq J_{t}\subseteq I_r$ and $J_{h}a\subseteq I_qa$, then $J_{h}a\subseteq  I_r\cap I_{q}a$. Therefore, by the property of reset left regular decompositions we get $I_{q}a\subseteq I_{r}$, whence $\varphi(\lambda(h,a))=r=\delta(\varphi(h),a)$, and this concludes the proof of the theorem.\qed
\end{proof}


\noindent \textbf{Theorem~\ref{theo: another restatement}.} \emph{Cerny's conjecture holds if and only if for any $\mathrsfs{B}\in\mathcal{L}(\Sigma)$  and $\rho\in\Cong_{k}(\mathrsfs{B})$ for all $k<\sqrt{\| \Syn(\mathrsfs{B})\|}+1$ we have
$$
\|\Syn(\mathrsfs{B}/\rho)\|<\|\Syn(\mathrsfs{B})\|
$$}
\begin{proof}
Since Cerny's conjecture holds if and only if it holds for strongly connected automata, we can suppose without loss of generality that the automata considered are strongly connected. Thus, suppose that Cerny's conjecture holds for strongly connected synchronizing automata and let $\mathrsfs{B}\in\mathcal{L}(\Sigma)$, $\rho\in\Cong_{k}(\mathrsfs{B})$ for some $k<\sqrt{\| \Syn(\mathrsfs{B})\|}+1$. Take $I=\Syn(\mathscr{B}/\rho)$. By Proposition~\ref{prop: class property} $\mathscr{B}$ is strongly connected, thus a quotient automaton $\mathscr{B}/\rho$ is strongly connected. Hence by Theorem~\ref{theo: Cerny rdc reform} we have $k\ge rdc(I) \ge \sqrt{\|I\|}+1$, i.e. $\|\Syn(\mathrsfs{B}/\rho)\|<\|\Syn(\mathrsfs{B})\|$.

Suppose that for any $\mathrsfs{B}\in\mathcal{L}(\Sigma)$ and $\rho\in\Cong_{k}(\mathrsfs{B})$ for all $k<\sqrt{\| \Syn(\mathrsfs{B})\|}+1$ the inequality in the statement of the theorem holds. Let $\mathrsfs{A}$ be a strongly connected synchronizing automaton with $k$ states. Let $w$ be a reset word for $\mathscr{A}$ of minimum length. For the word $w$ we build the automaton $\mathrsfs{B}$ of Theorem~\ref{theo: lifting property} associated to $w$ such that $\mathrsfs{A}$ is a homomorphic image of $\mathrsfs{B}$. Actually from the proof of Theorem~\ref{theo: lifting property} it follows that $\mathscr{A}$ can be viewed as a quotient automaton $\mathscr{B}/\rho$ for some $\rho\in\Cong_{k}(\mathrsfs{B})$. By the same theorem we also have $\Sigma^{*}w\Sigma^{*}\subseteq \Syn(\mathrsfs{B})\subseteq \Syn(\mathrsfs{A})$, hence
$$
\|\Syn(\mathrsfs{A})\|=|w|=\|\Syn(\mathrsfs{B})\|.
$$
Therefore, by the statement of the theorem we must have
$$
k\ge \sqrt{\| \Syn(\mathrsfs{B})\|}+1=\sqrt{\| \Syn(\mathrsfs{A})\|}+1
$$
whence $\mathrsfs{A}$ satisfies Cerny's conjecture.\qed
\end{proof}

\textbf{Proposition~\ref{prop: synt comp quotient}.}
\emph{Let $I=w^{-1}\Sigma^*w$ for some $w\in\Sigma^*$. The syntactic complexity of $I$ is equal to $$\sigma(I)=|w|+1+|\Pref(w)|+|\Fact(w)|+|\Suff(w)|-|\Pref_{syn}(w)|.$$}
\begin{proof}
Let $\mathscr{A}_w=\la P(w),\Sigma,\xi,q_n,\{q_n\}\ra$ be the minimal automaton recognizing~$I$ as in Proposition~\ref{prop: language charact. A_w}. So $P(w)=\{q_{0},\dots, q_{n}\}$ is the set of prefixes of the word $w$, $|q_i|=i$ for all indices $i$, and $\xi(q_{i},a)=(q_{i}a)\wedge_{s} w$ for any $q_i\in P(w)$, $a\in\Sigma$. By Proposition~\ref{prop: class property} $w$ is a reset word for $\mathscr{A}_w$ and $\mathscr{A}_w$ is strongly connected. Thus, since $w\in I$, we have $P(w)\dt w=\{q_n\}$. Furthermore, for each $q_i\in P(w)$ there exists some $u\in\Sigma^*$ such that $q_n\dt u=q_i$, hence $P(w)\dt wu=\{q_i\}$. Note that $|P(w)|=n+1$, so we can find $n+1$ reset words for $\mathscr{A}_w$ defining pairwise different transformations of the automaton.

Take any $u,v\in \Fact(w)$, $u\neq v$. There exist some $q_i,q_j\in P(w)$ such that $q_i\dt u=q_iu=q_{i+|u|}$ and $q_j\dt v=q_jv=q_{j+|v|}$ (see the illustration below).
$$\underbrace{\overbrace{w[1]\ldots w[i]}^{q_i}\overbrace{w[i+1]\ldots w[i+|u|]}^{u}}_{q_{i+|u|}}\quad\quad\quad
\underbrace{\overbrace{w[1]\ldots w[j]}^{q_j}\overbrace{w[j+1]\ldots w[j+|v|]}^{v}}_{q_{j+|v|}}
$$
Clearly, $q_0\dt u<_{p} q_{i+|u|}=q_i\dt u$ and $q_0\dt v<_{p} q_{j+|v|}=q_j\dt v$, hence $u$ and $v$ are not reset words for $\mathscr{A}_w$.
Without loss of generality suppose that $|u|\le |v|$. We show that $u$ and $v$ define different transformations of $\mathscr{A}_w$ by considering the following cases.

\textbf{Case 1.} Assume $|u|<|v|$. If $i=j$ then $q_i\dt u=q_{i+|u|}\neq q_{i+|v|}=q_i\dt v$. If $i<j$ then $q_j\dt u\neq q_j\dt v$ since $q_j\dt u=q_k$ for some $0\le k< n$ and $q_k<_{p}q_jv=q_j\dt v$. If $i>j$ then using an analogous argument we have $q_j\dt u\neq q_j\dt v$.

\textbf{Case 2.} Assume that $|u|=|v|$. If $i=j$ then $q_i\dt u=q_i\dt v$ since $|u|=|v|$. Thus $u=v$, which is a contradiction. If $i<j$ then $q_j\dt u\neq q_j\dt v$ since $q_j\dt u=q_k$ for some $0\le k< n$ and $q_k<_{p}q_jv=q_j\dt v$. If $i>j$ then using the same an analogous argument we have $q_i\dt u\neq q_i\dt v$.

Take any suffixes $s,t\in \Suff(w)$, $s\neq t$. There exist some $q_i,q_j\in P(w)$ such that $q_i\dt s=q_is=w=q_{n}$ and $q_j\dt t=q_jt=w=q_{n}$.
Without loss of generality suppose that $|s|\le|t|$. If $|s|=|t|$ then $s=t$, which is a contradiction. So we may assume $|s|<|t|$. Thus, $q_j\dt s\neq q_j\dt t$ since $q_j\dt s<_{p}w=q_j\dt t$. Therefore, different suffixes define different transformations of $\mathscr{A}_w$. Furthermore, $q_0\dt s\le_{p}s\neq w=q_i\dt s$, so $s\not \in\Syn(\mathscr{A}_w)$. Analogously, $t\not\in\Syn(\mathscr{A}_w)$. It remains to show that there is no proper suffix $t$ defining the same transformation of $\mathscr{A}_w$ as some inner factor different from $t$.
Let $u\in \Fact(w)$, $t\in \Suff(w)$ and $t\neq u$. Again, consider $q_i,q_j\in P(w)$ such that $q_i\dt u=q_iu$ and $q_j\dt t=w=q_n$. If $i\le j$ then, since $u\not\in \Suff(w)$, $q_i\dt u=q_iu<_{p}w=q_j\dt t$. If $i>j$ then $q_j<_{p}q_i$, thus $q_j\dt u<_{p}q_iu=q_i\dt u<_{p}w=q_j\dt t$. Hence, we get that $u$ and $t$ define different transformations of $\mathscr{A}_w$.

Take any prefixes $x,y\in \Pref(w)$, $x\neq y$. Since $q_0\dt x=x\neq y=q_0\dt y$, we get that $x$ and $y$ define different transformations of $\mathscr{A}_w$. We now show that proper prefixes define transformations which differ from transformations defined by any proper factor or a suffix. Indeed, take any two different words $x\in \Pref(w)$ and $u\in \Fact(w)$ such that $u\neq x$. If $|x|\ge |u|$ then $q_0\dt x=x\neq q_0\dt u$. If $|x|<|u|$ then $q_i\dt x\neq q_i\dt u$, where $q_i\dt u=q_iu$ for some $q_i\in P(w)$. Consider now any two different words $x\in \Pref(w)$ and $t\in \Suff(w)$ such that $t\neq x$. If $|x|\ge |t|$ then $q_0\dt x=x\neq q_0\dt t$ (otherwise $x$ would be a suffix of $w$). If $|x|<|t|$ then $q_j\dt x\neq q_j\dt t$, where $q_j\dt t=w$ for some $q_j\in P(w)$. If there is some $x\in\Pref(w)\cap\Pref_{syn}(w)$ then $x\in\Syn(\mathscr{A}_w)$, but each reset word for $\mathscr{A}_w$ belongs to the one of $n+1$ equivalence classes defined earlier. Therefore, we have proved $$\sigma(I)\geq |w|+1+|\Pref(w)|+|\Fact(w)|+|\Suff(w)|-|\Pref_{syn}(w)|$$ Now, take any $z\in\Sigma^*$ such that $z\not\in\Syn(\mathscr{A})\cup\Pref(w)\cup\Fact(w)\cup\Suff(w)$. It remains to show that $z$ does not define a new transformation of $\mathscr{A}_w$ that differs from transformations corresponding to reset words of $\mathscr{A}_w$, prefixes, suffixes or factors of $w$. Note that $z\neq w$ since $w\in\Syn(\mathscr{A}_w)$. If $|z|\geq |w|$ then $z\in\Syn(\mathscr{A}_w)$ since by Lemma~\ref{lem: charact wedge} we have $q_i\dt z=q_iz\wedge_s w=z\wedge_s w$ for all $q_i\in P(w)$, so we are done. Assume that $|z|<|w|$, and put $q_0\dt z=q_k$. The latter means that $q_k$ is the maximal suffix of $z$ which appears in $w$ as a prefix. We may assume that $q_i\dt z=q_j\neq q_k$ for some $q_i\in P(w)$, otherwise $z\in \Syn(w)$. By the definition of $q_k$ we have $q_k<_{s}q_j$. Furthermore, by the definition of $q_j$ we get $q_j<_{s}q_iz$. We have two cases: either $q_j\le_{s} z$, or $z<_{s}q_j$. If $q_j\le_{s} z$ then, since $q_{k}$ is the maximal suffix of $z$ which is also a prefix of $w$, we get $q_{j}\le_{s}q_{k}$, a contradiction.
In the second case we have $z<_f q_j$ then, since $q_j\le_{p} w$, we get $z<_{f}w$, a contradiction. So we have $\sigma(I)= |w|+1+|\Pref(w)|+|\Fact(w)|+|\Suff(w)|-|\Pref_{syn}(w)|$.\qed
\end{proof}

\end{document}